\def\draft{0}  
\documentclass[12pt]{article}
\usepackage[letterpaper,hmargin=1.5in,vmargin=1.5in]{geometry}
\usepackage{verbatim,sgame}

\usepackage{natbib}

\usepackage{tikz}
\usetikzlibrary{positioning,chains,fit,shapes,calc}

\usepackage{amsfonts,url,ifthen,epsfig,amsmath,amssymb,color,framed,hyperref,float}
\hypersetup{colorlinks=false}

\ifnum\draft=1
\newcommand{\Rnote}[1]{\begin{framed}\noindent \textcolor{red}{{#1}}\end{framed}} 
\else
\newcommand{\Rnote}[1]{}
\fi

\newcommand{\remove}[1]{}
\newtheorem{theorem}{Theorem}
\newtheorem{example}{Example}
\newtheorem{assumption}{Assumption}

\newtheorem{claim}{Claim}

\newtheorem{lemma}{Lemma}
\newtheorem{corollary}{Corollary}

\newtheorem{definition}{Definition}
\newtheorem{proposition}{Proposition}
\newtheorem{remk}[theorem]{Remark}

\def\FullBox{\hbox{\vrule width 8pt height 8pt depth 0pt}}

\def\qed{\ifmmode\qquad\FullBox\else{\unskip\nobreak\hfil
\penalty50\hskip1em\null\nobreak\hfil\FullBox
\parfillskip=0pt\finalhyphendemerits=0\endgraf}\fi}

\def\qedsketch{\ifmmode\Box\else{\unskip\nobreak\hfil
\penalty50\hskip1em\null\nobreak\hfil$\Box$
\parfillskip=0pt\finalhyphendemerits=0\endgraf}\fi}

\newenvironment{proof}{\begin{trivlist} \item {\bf Proof:~~}}
  {\qed\end{trivlist}}

\newcommand{\beq}{\begin{equation}}
\newcommand{\eeq}{\end{equation}}
\newcommand{\be}{\begin{enumerate}}
\newcommand{\ee}{\end{enumerate}}
\newcommand{\bi}{\begin{itemize}}
\newcommand{\ei}{\end{itemize}}
\newcommand{\bd}{\begin{description}}
\newcommand{\ed}{\end{description}}
\newcommand{\bc}{\begin{center}}
\newcommand{\ec}{\end{center}}
\newcommand{\bthm}{\begin{theorem}}
\newcommand{\ethm}{\end{theorem}}
\newcommand{\bdefi}{\begin{definition}}
\newcommand{\edefi}{\end{definition}}
\newcommand{\bcor}{\begin{corollary}}
\newcommand{\ecor}{\end{corollary}}
\newcommand{\blem}{\begin{lemma}}
\newcommand{\elem}{\end{lemma}}
\newcommand{\bexa}{\begin{example}}
\newcommand{\eexa}{\end{example}}
\newcommand{\bprop}{\begin{proposition}}
\newcommand{\eprop}{\end{proposition}}

\ifnum\draft=1
\newcommand{\ronote}[1]{\begin{framed}\noindent \textcolor{red}{{Ronen's note: #1}}\end{framed}} 
\newcommand{\ranote}[1]{\begin{framed}\noindent \textcolor{red}{{Rann's note: #1}}\end{framed}} 
\else
\newcommand{\ranote}[1]{}
\newcommand{\ronote}[1]{}
\fi

\newcommand{\eqdef}{\mathbin{\stackrel{\rm def}{=}}}

\def\real{\hbox{\rm\setbox1=\hbox{I}\copy1\kern-.45\wd1 R}}
\def\neal{\hbox{\rm\setbox1=\hbox{I}\copy1\kern-.45\wd1 N}}

\newcommand{\abs}[1]{\left|{#1}\right|}
\newcommand{\eps}{\varepsilon}

\newcommand{\ot}{\overline{t}}

\def\L{{\cal L}}
\def\M{{\cal M}}

\newcommand{\R}{{\mathbb R}}
\newcommand{\supp}{\mathrm{supp}}
\newcommand{\br}{\mathrm{BR}}

\begin{document}

\definecolor{myblue}{RGB}{80,80,160}
\definecolor{mygreen}{RGB}{80,160,80}

\begin{titlepage}
\title{Perception Games and Privacy%
\thanks{Gradwohl gratefully acknowledges the support of NSF award \#1216006. Smorodinsky gratefully acknowledges the support of ISF grant 2016301, the joint Microsoft-Technion e-Commerce Lab, Technion VPR grants and the Bernard M. Gordon Center for Systems Engineering at the Technion. We also thank Eddie Dekel, Jana Friedrichsen, Ehud Kalai, Gil Kalai, Birendra Rai, Juuso Valimaki and three anonymous reviewers for insightful comments.}}

\author{Ronen Gradwohl\thanks{Kellogg School of Management, Northwestern University, Evanston, IL 60208, USA. Email: \texttt{r-gradwohl@kellogg.northwestern.edu}.}
\and
Rann Smorodinsky\thanks{Faculty of Industrial Engineering and Management, The Technion -- Israel Institute of
Technology. Email: \texttt{rann@ie.technion.ac.il.}}}
\date{}
\maketitle

\begin{abstract}
Players have privacy concerns that may affect their choice of actions in strategic settings. We use a variant of signaling games to model this effect and study its relation to pooling behavior, misrepresentation of information, and inefficiency.
\end{abstract}

\thispagestyle{empty}

\end{titlepage}

\renewcommand{\thefootnote}{\arabic{footnote}}
\setcounter{footnote}{0}

\section{Introduction}
The concern for privacy has been a part of the popular debate over the past century and, with the prevalence of the Internet, even more so over the last decade. The relentless tracking of our activities, especially those online, puts forward a variety of new social dilemmas, concerns, and challenges, most of which fall under the general theme of privacy.
In this paper we single out one such challenge: to understand the effects of individuals' privacy concerns on decision making.
Can such concerns affect the choices made? Does the lack of privacy necessarily induce pooling behavior? And does this result in a loss of social efficiency?
To illustrate some of the issues involved, consider the following {\em stylized} example:%
\footnote{This example is inspired by Example 3.1 in \cite{nissim2012privacy}. Example \ref{Example1revisited} in the sequel is a rigorous reformulation of this example.}

\begin{example}\label{example:Coventry}
Alice holds one of two political views, denoted $\ell(eft)$ and $r(ight)$, and can subscribe to one (and only one) of two political blogs, denoted $L$ and $R$. Assuming Alice suspects no one is following her, she prefers to follow the blog that corresponds to her political views. We calibrate utilities so that reading the ``correct'' blog is worth 1, while reading the other blog is worth 0. In addition, Alice wants to maintain her political views concealed but is concerned that what she reads online may be monitored. This concern is captured by a disutility of 1 in case Alice's political views are completely exposed and a disutility of 0 if they are not exposed at all. Which blog should Alice follow?%
\footnote{One possible interpretation and story for Alice's concerns is that she does not want her future boss to know that she is liberal-leaning due to the possible conservative nature of her future boss, or vice-versa. Given the ambiguity the future holds she prefers that nothing is learned about her. This interpretation suggests that privacy concerns are a reduced-form model for some ambiguous future interaction.}
\end{example}

Observe the following two claims, which are independent of any reasonable definition of `privacy': (1) If Alice has a deterministic strategy that separates her two types (either she is always truthful or always lies), then her type can be fully deduced from her actions. As per our assumptions, this results in a non-positive utility for either types (note the additional disutility of 1 due to privacy concerns). (2) If Alice performs the same action for both types (possibly mixed), then nothing can be inferred from observing the blog she reads and thus her privacy is not jeopardized and no privacy-related disutility is incurred.

Is being truthful optimal? Being truthful leads to an award of 1 from the optimal action, but complete revelation of her political views and thus an additional disutility of 1. The total utility is therefore 0. On the other hand, any mixed type-independent strategy induces no privacy concerns and therefore leads to a positive expected utility for both types. Hence, it is superior to truthfulness. However, any such random behavior cannot be optimal as in the interim stage, once Alice forms her political views, deviation to the corresponding optimal blog is both undetectable and profitable.

 Through a formal reformulation of this example (see Example \ref{Example1revisited})  we later show that there are three equilibria in the model, of which truthfulness is actually one. The other two are pooling equilibria where both types pool on one of the blogs.

In the example above the individual is concerned with what others can learn about her type upon observing the action. This is common to a variety of settings where the action we take may inform others about who we are, what we like and dislike, what our weaknesses and strengths are, and so on. Thus, in choosing an action in a world where our actions are tracked, we may account for
what others might learn about us from observing the chosen action. Put differently, beyond the `material' implication of choosing an action, which is standard to utilitarian and game theoretic analysis, there is an extra `privacy' consideration from viewing actions as signals that inform others about the type of the individual.
This is the aspect of privacy we study in this paper.

To do so we consider a framework that accounts for privacy concerns by weaving it into the utility function, the notion of a strategy and the solution concept: First, an agent's utility function is comprised of a non-standard argument---a distribution over his type set---representing what others think of him. Second, strategies in our model are pairs composed of type dependent actions and perceptions, where the latter is a belief associated with each action. In other words, each type associates to every action a belief over the type set which he ascribes to other players. Finally, in equilibrium, perceptions must comply with rationality assumptions.

%
%
%
By adding particular structure to the extended utility function  we can use this new framework to capture privacy concerns. Admittedly, what exactly is meant by `privacy concerns' is quite elusive. To cope with this ambiguity we advocate a minimalist approach whereby we draw conclusions while assuming a minimal structure required for capturing privacy. In this vein, we introduce two alternative ways for capturing privacy:
\begin{itemize}
\item
A decision maker is {\em anonymity-seeking} if the optimal ex-post belief of others is equal to their ex-ante belief. In other words, any leakage of information is harmful.
\item
A decision maker is {\em identification-averse} if full revelation of his type is the worst possible outcome. In other words, maximal leakage of information is most harmful, privacy-wise.
\end{itemize}

With a model of a decision maker at hand we proceed to the analysis of what comprises an optimal action (subject to rationality arguments). Our first attempt at this considers a
setting with a single decision maker (DM)  who makes a decision in the presence of a passive observer we refer to as Big Brother (BB). The action taken by DM serves also as a signal to BB about the type of DM, and so we arrive at a model that is a variant of signaling games.%
\footnote{In the sequel, and especially in Appendix~\ref{sec:signaling games analogy}, we discuss the nuances that separate our model from the literature on signaling games.}
As in signaling games there is an inherent complexity in solving the DM's optimization problem due to the
 cross-dependence between actions and perceptions. Taking a certain action induces a certain perception, for which the given action might not be optimal. However, resorting to the new optimal action will change the induced perception, which may, in turn, render the new action inferior. To capture this interdependence, we consider a variant of signaling games that we call \textit{perception games}. A solution to a perception game, rather than being some optimal action, is an equilibrium notion composed of strategies and perceptions, where optimality and Bayesian consistency both play a role: given the strategy, the perception must be Bayesian-consistent with the prior subjective distribution; given the perception, the strategy must maximize expected utility.


\paragraph{Our contribution}


The main take-home message is in studying the implications of privacy concerns on action choice. The most intriguing question relates to the potential pooling behavior this induces. Thus, for both aforementioned notions of privacy concerns we introduce necessary and sufficient conditions for pooling behavior. The conditions are somewhat intuitive and are associated with the existence of a common undominated action---an action for which no other action is preferred, independently of perception. Thus, pooling prevails under a loose definition of privacy concerns.

Does this pooling phenomenon imply that the analysis of a decision problem is sensitive to the assumption on privacy concerns? In other words, if privacy concerns are ignored will we arrive at similar conclusions regarding the DM's choice of action? Framing this question from the perspective of revealed preferences we ask whether observing non-pooling behavior implies the absence of privacy concerns. By adding some additional structure on preferences we show that this conclusion is false.
In fact, fully-revealing behavior may be a result of a ``tyranny of the majority", where the coexistence of a minority with privacy concerns and a majority that has no such concerns ``forces" the former to nevertheless fully reveal their private information.
This is particularly relevant
in light of recent claims made in the public debate over privacy that, given how much of their private information people voluntarily share over social networks, they no longer care about privacy
 (see Section~\ref{sec:revealed-prefs}).

Finally, we consider the question of whether privacy legislation---in the form of shielding individual actions from public view---is beneficial to social
welfare. A longstanding critique by economists of such legislation is the potential damage it causes in terms of social welfare  \citep[e.g.,][]{posner1978economic}. A typical example is that hiding credit history necessarily leads to inaccurate risk assessment of lenders and so risky lenders could  face (relatively) low interest loans while confident lenders face high interest rates. Similarly, privacy legislation of past employment can lead to inefficient matching of employees and firms. However, this critique ignores the welfare implications derived
from the actual maintenance of privacy. In contrast, we put forward the intuition that if individuals have privacy concerns then enforcing the secrecy of their actions will lead to an increase
in utility. As the privacy critique mentioned before is a result of externalities (maintaining one's privacy may jeopardize someone else's utility) it follows that with only one individual privacy legislation increases welfare.  However, with more than a single player this is no longer true. We show by example that already with two individuals it is possible
that a \textit{lack} of secrecy actually leads to higher welfare for all individuals, even when those individuals have privacy concerns.


Overall, we view our contribution as twofold. On the one hand we introduce of abstract model and solution concept conducive to the study of privacy, and on the other we discuss
privacy and pooling.

\paragraph{Related literature}
There is a growing literature on privacy in various academic disciplines, such as law, philosophy, and computer science---see the recent surveys of \cite{solove2006brief,solove2011nothing}, \cite{nissenbaum2009privacy}, and  \cite{dwork2010differential}, respectively. Most of this literature does not treat the behavioral implications of privacy but rather discusses the pros and cons of privacy regulation.
Similarly, work on privacy in economics, much of which is surveyed in \cite{hui2006economics} and \cite{acquisti2015economics}, mostly examines the policy question of how the revelation of individuals' private information
affects welfare. In particular, privacy is related to knowledge and privacy regulations allows the information holder the advantage of leveraging the information to his best interest. In all these models information available to others is used in a very concrete way to the benefit of the other, and possibly used against the information holder whose privacy is analyzed. The players in such models do not have an inherent concern about privacy but rather view it as a tactical means for improving their utility.  In addition, the individuals have a clear understanding of the future interaction where such information could be used.

This paper, in contrast, examines privacy through the lens of an individual decision maker whose concern for privacy is somewhat different. The concern is either intrinsic, or, as in the above models, it is instrumental, but with the difference that the decision maker does not have a clear understanding of how the use of this information will play out in the future. To this end, our model accounts for privacy directly in the utility function, and studies the implications of this (we elaborate on this further in Section \ref{subsection percpetion vs esteem}). As such, this paper is methodologically closely related to various strands of the game theory literature and, in particular, to work on signaling games, psychological games, and social image.

In a signaling game \citep[see, e.g.,][]{sobel2009signaling}, an informed player sends a message to an uninformed player, the latter performs an action, and both obtain a reward that is a function of the information, the message, and the action.
Single-player perception games are closely related, except that the uninformed player takes no action but rather forms a belief that is Bayesian-consistent with the message and strategy of the informed player. However, if one replaces the belief formation of the uninformed player in our perception game with a strategic player who is asked to provide a prediction over the information held by the informed player and will be rewarded according to a proper scoring rule \citep[see][]{brier1950verification}, then one is back at a signaling game with a proper mapping of the equilibria.\footnote{The differences between perception games and signaling games are discussed in Appendix~\ref{sec:signaling games analogy}.
One main difference is that the latter typically assumes that there is some ideal type of DM whom all wish to emulate, whereas this does not hold in the former.}

Psychological games \citep{gilboa1988information,geanakoplos1989psychological,battigalli2009dynamic} are games in which the utilities of players depend not only on all players' actions,
but also on their beliefs about others' strategy profiles, as well as beliefs over such beliefs, and so on. Psychological games have been used to model emotions such as surprise, embarrassment, and guilt, among others.
Our paper is similar in that our model involves a decision maker whose utility depends on more than his action. Unlike psychological games, our model is rooted in a setting of incomplete information, and our focus is on the effect of privacy concerns.

Of course, the connection between an individual's choice of action and the resulting inferences about the individual's fundamental traits goes beyond privacy, and
there is a vast literature that examines these other connections.
Research on conspicuous goods, for example, acknowledges that such goods provide a dual benefit---a direct consumption effect as well as an indirect effect due to the signal the consumption of these goods send (e.g., an
individual enjoys driving his new hybrid Toyota Prius, as it is a quality car and, in addition, he is perceived by others as environmentally friendly, which he is happy about). The literature on conformity similarly studies situations where persistent norms emerge, as opposed to transient ones (fads), in a society where perceptions matter. Finally, the literature on self-image also connects action and perception.%
\footnote{A very partial and
incomplete list of relevant references includes  \cite{bernheim1994theory}, \cite{glazer1996signaling}, and
\cite{ireland1994limiting} on conformity, charity, and status, respectively; \cite{benabou2006incentives} on pro-sociality; \cite{becker1974theory} and, more recently, \cite{rayo2013monopolistic} and \cite{friedrichsen2013image} on self-image. See the references therein for further related
literature.}
One particular paper which connects esteem with privacy policies (a policy that hides agents' actions) is that of \cite{daughety2010public}.
\citeauthor{daughety2010public}'s agent has utility over four terms: type, a consumption good, her level
of provision of a public good (or bad), and an esteem term. The esteem term is updated
when other agents observe the active agent's level of provision of the public good (or bad).

In the context of online advertising \cite{de2016online} study the implications of imposing a privacy policy in the market for targeted advertising. This is done by comparing the market outcome in two scenarios: with and without imposing a privacy policy. In that work, however, the individuals whose privacy is protected play no strategic role, whereas in out work we focus on the strategic implications of privacy concerns.

Privacy concerns are quite different from esteem and issues of self image in that the latter two typically involve an ordered set of types where all agents would like to be perceived as `high' types. In particular, agents that have high types prefer that their type be known. In contrast, our work focuses on privacy concerns where  agents of all types appreciate the anonymity of taking an action, assuming all else if equal. We elaborate further of the differences between social image modeling and privacy concerns in Section \ref{subsection percpetion vs esteem}.

Repeated games provide a tool for studying privacy policies. One can study and compare the equilibrium implications between two settings, when actions at early stages are observed and used to determine players' action at later stages vs.\ the case they are not (say, due to some privacy policy). In the repeated game setting equilibrium play is different if actions are monitored by players or not. The roots of this approach and the implications of monitoring can be traced back to the monumental work of \cite{aumann1995repeated}, who study an abstract zero-sum setting. In recent years this has been studied in more specific economic models such as monopolist pricing \citep{taylor2004consumer}, sequential contracting \citep{calzolari2006optimality}, repeated signaling games \citep{chen2014privacy} and more (we refer the interested reader to a survey of \cite{mailath2006repeated}). In this line of work, perceptions play no direct role in determining a player's utility, but rather an indirect role through the stream of future payoffs.
The way perceptions play a role in our model can be interpreted as a reduced-form model of a repeated interaction where the future evolution of a game is ambiguous, where ambiguity may be with respect to the nature of the opponents, the possible set of actions, the timing or the payoff function.

One paper in which players also have reduced-form preferences for privacy is that of \cite{cummings2016empirical}, who study the testable implications of choice data when consumers
have such preferences. While related to this paper, especially to Section~\ref{sec:revealed-prefs}, the work of \citeauthor{cummings2016empirical} is very different from ours. Most
significantly, in their model a consumer's choices reveal his preferences, and, unlike in our paper, there is no equilibrium involving the observer's inferences and optimal choices given those inferences.

\paragraph{Organization} The rest of the paper is organized as follows. Section~\ref{sec:model} introduces perception games and
states an equilibrium existence result. Section~\ref{sec:privacy} contains our analysis of the effect of privacy concerns on decision-making: Section~\ref{sec:pooling}
on the existence of a pooling equilibrium, Section~\ref{sec:revealed-prefs} on whether pooling is necessarily implied by privacy concerns, and Section~\ref{sec:general-model}
on the effect of privacy concerns on social efficiency. Section~\ref{sec:conclusion} concludes.

\section{Perception Games}\label{sec:model}

A decision maker (DM) must choose an action from some finite set of possibilities, $A$. The utility of the DM is derived from his chosen action, his type, and the beliefs others have of him. 
The DM's chosen action is observed by a second, inactive player, referred to as Big Brother (BB), who forms a belief over the type set. 

Formally, a {\it perception game} is a tuple $(T,A,\beta,u)$, where
\begin{itemize}
\item
$T$ is a finite type space for the active player (the DM);
\item
$A$ is a finite action space (for the DM);
\item
$\beta \in \Delta(T)$
is the prior belief of BB over the set of types;
and
\item
$u:T \times A \times \Delta(T) \to \R$ is the utility function of the DM.
\end{itemize}

The distinguishing feature of this model is the introduction of elements in $\Delta(T)$ as arguments of the utility function. An element in $\Delta(T)$, a distribution over the type set $T$, is the DM's perception: the belief he ascribes to BB about his type. 
See Section~\ref{sec:modeling} for comments on this modeling choice.

A strategy of the active player is a vector $\sigma= \{\sigma_t\}_{t\in T}$, where $\sigma_t \in \Delta(A)$
is the mixed strategy of type $t$.
A \textit{perception} for type $t$ of the active player is a function $\tau_t: A \to \Delta(T)$. That is, for any type $t\in T$ and action $a \in A$, let $\tau_t(a)$ denote the posterior distribution over $T$ following the action $a$. The DM's perception is a vector of perceptions, $\tau = \{\tau_t\}_{t\in T}$, one for each type. The expected payoff of the active player at type $t$, given the perception $\tau_t$ and the strategy $\sigma_t\in\Delta(A)$, is
$$
U(t,\sigma_t,\tau_t)  =
\sum_{a\in A}\sigma_t(a)\cdot u(t,a,\tau_t(a)).
$$

\subsection{A perception equilibrium}

Ignoring perception, the DM's choice of action is determined by maximizing his expected utility, for any given type. In particular, the optimal action for one type would be independent of that of another type. However, when accounting for perception, this notion of maximization is misleading. To see this, note that the utility of some action taken by the DM of one type depends on the belief of BB about the type of DM, updated after seeing the latter's action. However, this belief is derived from the possibility that the same action could have been taken by a different type of the DM, and so there is a clear externality among types. To capture this we couple optimization and beliefs into the notion of a {\em perception equilibrium}.\footnote{The externality across types and our notion of perception equilibrium is analogous to Perfect Bayesian Equilibria in signaling games---see Section~\ref{sec:signaling games analogy}
for an explicit comparison and discussion.}

We precede the definition by some notation.
Let $P(t,a | \sigma) =
\beta(t) \cdot \sigma_t(a)$
be the probability that BB assigns to type $t$ and action $a$, conditional on her playing the strategy profile $\sigma$.
Then $P(a|\sigma) = \sum_{t\in T}P(t,a|\sigma)$ is the probability she assigns to action $a$.
Let $A_\sigma= \{\cup_{\tilde t \in \supp(\beta)}\supp(\sigma(\tilde t))\}$, where $\supp(\sigma(\tilde t))$ is
the set of actions taken by type $\tilde t$ under strategy $\sigma$. The set $A_\sigma$ is then the set of all actions that can possibly be
taken by some possible type under $\sigma$. By definition, $P(a|\sigma)>0$ for all actions in $A_\sigma$, so
$P(t| a,\sigma)= \frac{P(t,a|\sigma)}{P(a|\sigma)}$, the conditional probability that BB assigns to the DM being of type $t$ upon seeing action $a$, is well defined for all $a \in A_\sigma$.

\begin{definition}
A strategy--perception pair $(\sigma,\tau)$ is {\em consistent} if $\tau_t(a)\left(\bar t\right) = P\left(\bar t|a,\sigma\right)$ for all $t, \bar t \in T$  and
$a \in A_\sigma$.
\end{definition}

In words, a player's strategy--perception pair is consistent if the perception is derived by Bayes' rule from the prior and strategy whenever possible.

\bdefi
A {\em perception equilibrium} is a consistent strategy--perception pair $(\sigma,\tau)$ such that
$\sigma$ is a best-reply profile: for every $t\in T$, it holds that $U(t,\sigma_t,\tau_t) \ge
U(t,\hat\sigma,\tau_t)$ for all $\hat\sigma \in \Delta(A)$.
\edefi

Our modeling choice allows for different types to hold different perceptions. This flexibility is economically interesting, 
as different types of DM may have a different BB in mind. In Example~\ref{example:Coventry}, for instance,
it may be the case that a left-leaning Alice has in mind a right-leaning boss whereas the opposite is true for a right-leaning Alice.
Note, however, that despite this flexibility, in equilibrium perceptions must be consistent with Bayes' rule, which implies that perceptions are equal following those actions that are sent with positive probability. However, off the equilibrium path---following actions which no type takes with positive probability---perceptions may differ.%
\footnote{This gives the perception equilibrium some flavor of a subjective or self-confirming equilibrium \citep[see][]{kalai1993subjective, fudenberg1993self}.}

In the Appendix we extend our model beyond the single-player setting. In the extended version of perception games each player serves as some Big Brother for other players and, as such, can be one of many types. We consider a model where each type's beliefs over others' types is subjective and we similarly depart from the common prior assumption. In contrast with the single-player setting, in the extended setting the Bayesian consistency assumption will not reduce ex-post perceptions to be same for all types, even following positive probability events.


As previously discussed, perception games are a variant of signaling games. However, although the DM (``sender" in the jargon of signaling games) has finite type and action spaces,  BB (``receiver") has an infinite action space, as her belief over DM could be any probability distribution. Thus, it is not clear whether equilibrium existence results in the literature on signaling games ensure equilibrium existence here. In fact, to ensure existence we must add the following technical assumption on the utility function:%
\footnote{A further discussion of related equilibrium existence results in signaling games is provided in Appendix \ref{sec:signaling games analogy}.}
\begin{assumption}\label{assumption:continuity}
$u:T  \times A \times \Delta(T) \to \R$ is continuous in the last argument.%
\footnote{One must be careful with this seemingly innocuous assumption. Note that if we view the perception game as
a reduced form of some repeated interaction where BB is expected to take some future action out of a
finite set of actions, then this action will change discretely
at certain beliefs, and therefore the DM''s utility will be discontinuous as a
function of BB''s beliefs.}
\end{assumption}

\begin{proposition}\label{prop_existence}
A perception equilibrium exists in any perception game that satisfies Assumption \ref{assumption:continuity}.
\end{proposition}

 The proof follows from standard arguments and is thus relegated to Appendix~\ref{sec:proofs}.

\begin{remk}
Assumption~\ref{assumption:continuity} is necessary for the equilibrium existence result.\footnote{Furthermore, in Appendix~\ref{sec:continuity-necessary} we show that
upper or lower semicontinuity alone are not sufficient.} We demonstrate this with an example in Appendix~\ref{sec:continuity-necessary}.
\end{remk}


\section{Perception and Privacy}\label{sec:privacy}

Recall that  perception games were introduced as a tool to study the implications of privacy concerns. In the context of privacy one should think of the set of types as the traits individuals would like to keep private. As an example, an element $t \in T$ can represent the DM's political views, his attitude towards risk,
his location, his willingness-to-pay for some good, etc.

In order to capture the possible predilection for privacy we proceed by introducing some structure on the utility function $u$.
%
%
%
The study of perceptions in the literature, either those of society or self-image, typically assume that types have some natural order and all agents would like to be perceived as (close as possible to) the ideal type (e.g., benevolent, heroic, generous, romantic, and so on). In contrast, in the context of privacy an alternative concern about perceptions comes to mind. Privacy may be thought of as the \emph{dis}interest of an individual to have others learn about him. Of course, as already discussed in the introduction, what exactly is meant by `privacy concerns' is quite elusive. Hence, we introduce two alternative structures to capture a predilection for privacy: the DM could be anonymity-seeking or identification-averse.

The first conception of privacy concerns is that agents optimally do not wish anything to be disclosed about them, beyond what is already known.
This is captured by the following definition:
\bdefi
\label{upper_def_privacy}
A DM of type $t\in T$ in a perception game $(T,A,\beta,u)$ is {\em anonymity-seeking}  if for every $a\in A$,
$$\beta \in \arg\max_{\mu\in\Delta(T)} u(t,a,\mu).$$
A DM is {\em anonymity-seeking} if he is anonymity-seeking for every type $t\in T$.
\edefi

In words, anonymity-seeking preferences captures the setting where type $t$ prefers that nothing be disclosed about his type beyond BB's initial belief,
and who thus wishes to be perceived as an ``average" type.%
\footnote{A possible critique of an anonymity-seeking preference is that it exhibits some logical inconsistency in a dynamic setting. Assume a DM faces two consecutive perception games, and assume that there are strategies for both games such that BB does learn about the DM in each, but in a way that cancels out. That is, following both games, BB's perception is back at its starting point. In the grand game (composed of the two stage games), the DM lands the best perception yet in each separate stage he does not. Note
that identification-aversion does not suffer from this inconsistency.}
For example, an employee may wish to keep his political views private from a future employer, and a consumer of online adult entertainment, even when known to be a consumer of such entertainment, would prefer to maintain ambiguity over the preferred specific genre.

A different conception of privacy concerns is that agents do not wish their true type to be revealed. This is captured in the following definition:
\bdefi
\label{lower_def_privacy}
A DM of type $t\in T$ in a perception game $(T,A,\beta,u)$ is {\em identification-averse}  if for every $a\in A$,
$$\chi(t) \in \arg\min_{\mu\in\Delta(T)} u(t,a,\mu),$$
where $\chi(t)$ is the Dirac measure on $t$.
A DM is {\em identification-averse} if he is identification-averse for every type $t\in T$.
\edefi
Unlike an anonymity-seeking DM, an identification-averse one may benefit from the leakage of some (correct or incorrect) information to BB. However, he
does not wish his true type to be disclosed.
For example, an individual worried about stalking wishes to keep his location private, and a consumer contemplating a purchase does not want to disclose
his willingness-to-pay.

Note that for both variants of privacy concern the requirement is that the concern holds for all types. This is necessary for Theorems \ref{thm:pooling-upper} and \ref{thm:pooling-lower}, which do not hold if we require privacy concerns for only one type.

We now turn to study the implications of privacy concerns. How do privacy concerns factor into the choice of action? Does the DM pool actions across types in order to masquerade his type?  Can privacy be detected from revealed preferences? Is the outcome necessarily inferior compared to the setting where the actions of the DM are not observed?

\subsection{Privacy and pooling}
\label{sec:pooling}

We begin the study of privacy and pooling by reformulating Example \ref{example:Coventry} more rigorously:

\begin{example}\label{Example1revisited}
Alice is one of two types $\{\ell,r\}$, with a prior $\beta(\ell)=\beta(r) = 0.5$, and she must choose between one of two actions $\{L,R\}$. The utility of Alice is given by
$u(a,t,p)= \mathbf{1}_{\{t=a\}}- 2|p-0.5|$, where $\mathbf{1}_{\{t=a\}}$ is the indicator function for the event $\{t=a\}$ and $p$ is the probability assigned to Alice being of type $\ell$. One can easily verify that Alice is both anonymity-seeking and identification-averse.

Consider the pooling decision where Alice reads blog $R$ no matter what her type is, but believes that if blog $L$ would have been read then $BB$ would think Alice is of type $\ell$. Note that if both types follow this strategy then neither incurs any disutility due to perceptions. In addition, type $r$ gets a `material' utility of 1, whereas type $\ell$ gets a `material' utility of $0$. More importantly, neither type has an incentive to deviate, so this defines a perception equilibrium. Symmetrically, there exists another perception equilibrium, where actions $R$ and $L$ are reversed.
A third perception equilibrium is one where both types choose their favorite blog. In this equilibrium both types obtain a utility of 0.\footnote{In this example there
is no equilibrium in which the DM uses a type-dependent lottery (in any such strategy profile type $\ell$ could profitably deviate to always
play $L$
and type $r$ to always play $R$). However, this need not always be the case, unlike in standard signaling games.}


\end{example}

In fact, the pooling behavior in this concrete example is an instance of a broader phenomenon. In the following two theorems we describe necessary and sufficient
conditions for {\em full pooling}, in which all types take the same action. The first theorem applies to an anonymity-seeking DM, and the second
to an identification-averse one.

We begin with a definition:
\begin{definition}\label{def:potentially-optimal}
An action $a\in A$ is {\em potentially optimal} for type $t\in T$ if for any action $a'$ there exists a pair of beliefs $\mu, \mu'$ such that $u(t,a,\mu) \geq u(t,a',\mu')$.
Let $L(t)$ be the set of actions that are potentially optimal.
\end{definition}
That is, an action $a$ is potentially optimal if there is no other action that is strictly better for all perceptions. The following theorem provides necessary and sufficient conditions
for full pooling, for an anonymity-seeking DM:\footnote{In fact, under a slightly stronger definition of anonymity-seeking---that $\beta$ be the unique element of the $\arg\max$
in Definition~\ref{upper_def_privacy}---the pooling equilibria, when they exist, may have the additional property that they are not dominated by any other 
consistent strategy-perception pair. Such a pooling equilibrium can be constructed by choosing  $a\in \L$ such that for some type $t\in T$ it holds that 
$a\in \arg\max_{a'\in A} u(t,a',\beta)$. Such an $a$ always exists, and in this equilibrium type $t$ obtains his strictly highest possible payoff.}

\begin{theorem}\label{thm:pooling-upper}
For an anonymity-seeking DM, full pooling is an equilibrium if and only if $\bigcap_t L(t) \neq \emptyset$.
\end{theorem}

The intuition underlying Theorem~\ref{thm:pooling-upper} is quite simple: If the condition is not satisfied then each action is dominated for at least one of the types. 
A dominated action is never played, and so full pooling cannot occur on any action. Sufficiency of the condition hinges on the flexibility the model regarding off-equilibrium perceptions.
 The formal proof follows: 


\begin{proof}
Let $\L=\bigcap_t L(t) \neq \emptyset$, and suppose first that $\L\neq \emptyset$.
Fix an action $a\in \L$ 
, and for each $t\in T$ and $a'\in A\setminus\{a\}$ let $\tau_t^{a'}\in\Delta(T)$ be a distribution for which $u(t,a,\tau_t^{a})\geq u(t,a',\tau_t^{a'})$
for some $\tau_t^{a}$.
Such a $\tau_t^{a'}$ exists since $a$ is potentially optimal for $t$.
Since the DM is anonymity-seeking, it holds that $u(t,a,\beta)\geq u(t,a,\tau_t^{a})\geq u(t,a',\tau_t^{a'})$.

Let us now construct the perception $\tau_t$ that supports action $a$ in a fully pooling equilibrium: $\tau_t(a)=\beta$ and $\tau_t(a')=\tau_t^{a'}$ for all $a'\in A\setminus\{a\}$. Denoting by  $\sigma^a$  the pure strategy $a$ the above yields
$U(t,\sigma^a,\tau_t)\geq U(t,\bar\sigma,\tau_t)$ for any alternative strategy $\bar\sigma$.
Letting $\sigma$ be the pure strategy
for which $\sigma(t)=\sigma^a$ for all $t\in T$ and $\tau=\{\tau_t\}_{t\in T}$, we obtain $(\sigma,\tau)$ as the desired fully-pooling equilibrium.

Next, assume $\L=\emptyset$. Suppose towards a contradiction that there does exist a fully-pooling equilibrium $(\sigma,\tau)$: for some $a\in A$, the strategy profile $\sigma$
satisfies
 $\sigma(t)=\sigma^a$ for all $t\in T$.
Since $\L=\emptyset$, there exists a type $t\in T$ for whom $a$ is not potentially optimal. That is, there exists some action $a'\neq a$ such that for all
pairs of beliefs $\mu, \mu'$ it holds that $u(t,a,\mu) < u(t,a',\mu')$. In particular, this means that  $u(t,a,\beta) < u(t,a',\tau_t(a'))$ (where $\tau_t(a')$ is the perception following action $a'$ in the alleged equilibrium). Thus, a deviation to action $a'$ is
profitable for type $t$, contradicting the assumption that $(\sigma,\tau)$ is an equilibrium.
\end{proof}

Next, we state a similar theorem for an identification-averse DM. Again, we first need a definition, a variant of Definition~\ref{def:potentially-optimal}:
\begin{definition}
An action $a\in A$ is {\em potentially optimal at the extremes} for type $t\in T$ if for any action $a'$ it holds that $u(t,a,\beta) \geq u(t,a',\chi(t))$.
Let $M(t)$ be the set of actions that are potentially optimal at the extremes.
\end{definition}

Given this definition, the following theorem provides necessary and sufficient conditions for full pooling, for an identification-averse DM:

\begin{theorem}\label{thm:pooling-lower}
For an identification-averse DM, full pooling is an equilibrium if and only if $\bigcap_t M(t) \neq \emptyset$.
\end{theorem}

The intuition underlying  Theorem~\ref{thm:pooling-lower} is similar to that of Theorem~\ref{thm:pooling-upper}.
\begin{proof}
Let $\M=\bigcap_t M(t) \neq \emptyset$, and suppose first that $\M\neq \emptyset$.
Fix any action $a\in \M$, and for each $t\in T$ and $a'\in A\setminus\{a\}$ let $\tau_t(a')=\chi(t)$, the Dirac measure on $t$.
Since the DM is identification-averse, it holds that $u(t,a,\beta)\geq u(t,a',\chi(t)) = u(t,a',\tau_t(a'))$. This implies that
$U(t,\sigma^a,\tau_t)\geq U(t,\sigma^{a'},\tau_t)$,
where $\tau_t(a)=\beta$. Letting $\sigma$ be the pure strategy
for which $\sigma(t)=\sigma^a$ for all $t\in T$ and $\tau=\{\tau_t\}_{t\in T}$ we obtain $(\sigma,\tau)$ as the desired fully-pooling equilibrium.

Next, assume $\M=\emptyset$. Suppose towards a contradiction that there does exist a fully-pooling equilibrium $(\sigma,\tau)$: for some $a\in A$, the strategy
 profile $\sigma$ satisfies $\sigma(t)=a$ for all $t\in T$.
Since $\M=\emptyset$, there exists a type $t\in T$ for whom $a$ is not potentially optimal at the extremes. That is, there exists some action $a'\neq a$ such that
$u(t,a,\beta) < u(t,a',\chi(t))$. Furthermore, since the DM is identification-averse, $u(t,a',\chi(t))\leq u(t,a',\mu)$ for all $\mu\in\Delta(T)$.
In particular, this means that  $u(t,a',\chi(t))\leq u(t,a',\tau_t(a'))$, and so $u(t,a',\beta)\leq u(t,a',\tau_t(a'))$. Thus, a deviation to action $a'$ is
profitable for type $t$, contradicting the assumption that $(\sigma,\tau)$ is an equilibrium.
\end{proof}

The pooling phenomena of Theorems \ref{thm:pooling-upper} and \ref{thm:pooling-lower}, which are possibly inconsistent with material considerations, are reminiscent of pooling phenomena in the literature on self-image. However, the underlying forces that induce pooling are slightly different. In the literature on self-image players pool because they all want to be perceived as some ideal type. In our context an ideal type need not exist, and pooling is a result of players' desire not to disclose information about their type, even if different players agree on which perception is best.

Recall that in Example \ref{Example1revisited}, beyond the pooling equilibria, there is a non-pooling equilibrium where both types read their optimal blog. This demonstrates that privacy concerns need not imply pooling, nor do they imply that the DM will forego an optimal action. This raises the following natural question: Can we determine whether or not individuals have privacy concerns from observing their actions?
This question is explored further in Section~\ref{sec:revealed-prefs}.

\subsection{Privacy and revealed preferences}\label{sec:revealed-prefs}
The increasing prevalence of social and economic endeavors online has prompted some to declare that individuals no longer care about privacy.
For example, Facebook founder Mark Zuckerberg has defended a change in Facebook's privacy policy---switching the default option for the visibility of posts from
private to public---by stating that
``People have really gotten comfortable not only sharing more information and different kinds, but more openly and with more people. That social norm is just something that has evolved over time'' \citep[see, e.g.,][]{kirkpatrick2010facebooks}.
In other words, the observed activity in social networks reflects, according to Zuckerberg, a change in social norms to a reduction in concerns for privacy. Can this deduction be sustained by our theory?

For the rest of this section we will consider a simplified version of our model, in which utilities are additive: Formally, for each type $t$,
action $a$, and perception
$\mu$ let
$$u(t, a, \mu) = v(t,a)-w(t,\mu),$$
where the first summand is type $t$'s utility from action $a$ and  the second summand is the privacy loss of that type associated with a perception $\mu$.
Furthermore, we will assume that privacy concerns are separable:
suppose each type consists of two elements, an outcome-type and a privacy-type, namely $T=T_o\times T_p$.
For each type $t=(t_o,t_p)$, we will make two simplifying assumptions:
\begin{enumerate}
\item The utility from actions depends only on the outcome-type: $v(t,a)=v(t',a)$ whenever $t_o=t_o'$;
\item The disutility from perceptions is a function of perceptions over the outcome-type only:
$w(t,\mu)=w(t,\mu')$ whenever $\mu|_{T_o}=\mu'|_{T_o}$, where $\mu|_{T_o}$ is the marginal distribution of $\mu$ over $T_o$.
\end{enumerate}
Finally, for each $t\in T$ let $a_t = \arg\max_a v(t,a)$, and suppose that $a_t\neq a_{t'}$ whenever $t_o\neq t'_o$.

In the context of social networks, these assumptions seem reasonable. As an example, suppose the outcome-type corresponds to the player's location, the action corresponds
to sharing a location online, and $v(t,a)$ corresponds to the player's
utility from posting a location and sharing it with his friends online (what is known as a `check-in' on Facebook). Privacy-types correspond to the player's level of concern about possible leakage of his location to other parties
and misuse of this information. Note that players derive utility from posting their location, but not from posting their level of privacy concern, and so $v$ depends only on $t_o$ and
$a$. Additionally, privacy concerns are with respect to a third party's perception of the player's location---that is, his outcome-type---and not their perception of
his level of privacy concern.

Now, in this simplified setting, Theorems~\ref{thm:pooling-upper} and~\ref{thm:pooling-lower} provide necessary and sufficient conditions for full pooling. But is {\em full separation},
in which each outcome-type
posts a unique location, possible? Obviously, if $w\equiv 0$ then the answer is positive, since then each type $t$ will post location $a_t$. But what if players do
have privacy concerns?

The following assumption places a lower bound on utilities from optimal actions. In the social network example, it states that the benefit of posting a true location online
over posting a false one
is bounded below by the difference in privacy loss associated with the true location revealed versus a false location revealed.
\begin{assumption}\label{assumption:opt-actions}
For each pair of types $t=(t_o,t_p)$ and $t'=(t'_o,t'_p)$ with $t_o\neq t'_o$ it holds that $v(t,a_t)-v(t,a_{t'}) > w(t, \chi(t)) - w(t, \chi(t'))$, where $\chi(t)$ is the Dirac measure on $t$.
\end{assumption}

Note that Assumption~\ref{assumption:opt-actions} is compatible with a DM being both anonymity-seeking and identification-averse.


A first observation is that under Assumption~\ref{assumption:opt-actions}, there exists a separating equilibrium, in which each type $t$ takes action $a_t$. This is related to
the similar observation made in Example \ref{Example1revisited}.
\begin{claim}
In any perception game satisfying Assumption~\ref{assumption:opt-actions} there exists a perception $\tau$ such that $(\sigma,\tau)$ is an equilibrium, and where $\sigma(t)=\sigma^{a_t}$ for
all $t\in T$.
\end{claim}

\begin{proof}
For each $t\in T$ and $a\in A_{\max} = \{a: \exists t\in T \mbox{ s.t.\ } a=a_t\}$, let $\tau_t(a)=\beta|_{\{t:a_t=a\}}$, where $\beta|_{\{t:a_t=a\}}$ is the distribution $\beta$ conditional
on the type belonging to the set
$\{t:a_t=a\}$.  For each $t\in T$ and $a\not\in A_{\max}$, let $\tau_t(a)=\chi(t)$.
It is straightforward to verify that $(\sigma,\tau)$ is consistent.

To see that $\sigma$ is a best-reply profile, fix some type $t$ and observe that the utility of type $t$ under $(\sigma,\tau)$ is $v(t,a_t)-w(t,\beta|_{\{t':a_{t'}=a_t\}})$.
Note also that $w(t,\beta|_{\{t':a_{t'}=a_t\}}) = w(t,\chi(t))$, since under $\sigma$ all types taking action $a_t$ have the same outcome type $t_o$, and the
disutility from perceptions is a function of perceptions over the outcome-type only.

Now consider a possible deviation of type $t$ from $\sigma^{a_t}$ to $\sigma^{a}$, where $a\neq a_t$. If $a\in A_{\max}$ then $\tau_t(a)=\beta|_{\{t':a_{t'}=a\}}$.
Note that, as above, the disutility $w(t,\beta|_{\{t':a_{t'}=a\}}) = w(t,\chi(t''))$, where $t''$ is some type for whom $a_{t''}=a$. Thus, under this deviation,
$t$'s utility is $v(t, a_{t''})-w(t,\chi(t''))$. Assumption~\ref{assumption:opt-actions} implies that this deviation is not profitable.
On the other hand, if the deviation $\sigma^a$ is such that  $a\not\in A_{\max}$ then $t$'s utility under the deviation is $v(t,a)-w(t,\chi(t))$. However, $v(t,a_t)\geq v(t,a)$ and so
$v(t,a_t)-w(t,\chi(t))\geq v(t,a)-w(t,\chi(t))$, implying that again the
deviation is not profitable. Thus, there is no profitable deviation, and so $\sigma$ is a best-reply profile.
\end{proof}

So under Assumption~\ref{assumption:opt-actions} there always exists a separating equilibrium, even if the DM has privacy concerns (he is anonymity-seeking or identification-averse).
What makes Assumption~\ref{assumption:opt-actions} interesting, however, is that it may entail separation as the {\em unique} equilibrium.
Suppose $\abs{A}=\abs{T_o}$, and so every action is optimal for some outcome-type.
Let $\overline{t}\in T_p$ be a privacy-type who does not care about perception:
formally, $w(t,\cdot)\equiv 0$ whenever $t_p = \overline{t}$. Other privacy-types may have {\em arbitrary} privacy concerns. Let $\alpha$ be the fraction of types $t\in T$ (under the prior $\beta$)
with $t_p=\overline{t}$. Then:
\begin{theorem}\label{Claim:robustness}
For any perception game satisfying Assumption~\ref{assumption:opt-actions} there exists an $\alpha_0<1$ such that if $\alpha\geq \alpha_0$ then full separation is the unique equilibrium.
\end{theorem}


\begin{proof}
%
%
%
For each action $a$, let $T_a\subseteq T$ be the set of types for whom $a$ is the optimal action, namely $T_a=\{t\in T: a_t=a\}$.
Let $\alpha$ be large enough so that $P_{t'\sim \beta}(t'\in T_a \cap t'_p=\ot)>0$ for every $a\in A$: in words, for every action $a$ there is a positive
measure of types who do not care about perception and for whom $a$ is optimal.
This will be true, for example, whenever $\alpha > \max_a 1-P(t'\in T_a)$.

Consider some action profile $\sigma$ in which types who do not care about perception play optimally, namely $\sigma(t)=\sigma^{a_t}$ whenever $t_o=\overline{t}$.
For any type $t$, let $\tau_{t}^o(a)$ be the marginal distribution of the perception on action $a$ over $T_o$.
Let $\tau_{t}^o(a)'$ denote the probability that $\tau_{t}^o(a)$ assigns to the event $(t'\in T_a)$. 
Then
$$\tau_{t}^o(a)'%
=\frac{P(t'\in T_a \cap t'_p=\ot) + \sum_{t':t'\in T_a \cap t'_p\neq \ot}P(t')\cdot\sigma(t')(a)}{P(t'\in T_a \cap t'_p=\ot) + \sum_{t':t'_p\neq \ot}P(t')\cdot\sigma(t')(a)},$$
where the probabilities are over the choice of $t'$ according to the prior $\beta$.
Observe that $\sum_{t':t'_p\neq \ot}P(t')=1-\alpha$, and so
 as $\alpha\rightarrow 1$ both $\sum_{t':t'_p\neq \ot}P(t') \rightarrow 0$ and $\sum_{t':t'\in T_a \cap t'_p\neq \ot}P(t')\rightarrow 0$. This
 implies that the second summand of both the numerator and denominator in the calculation of $\tau_{t}^o(a)'$ approach 0, and so  $\tau_{t}^o(a)'\rightarrow 1$.

For any type $t$ with $t_p\neq \ot$, the utility at an action $a$ is $v(t,a)-w(t,\tau_t(a))$. From continuity, the fact that the disutility from perception
depends only on the perception over the outcome type, and the fact that the probability assigned to the outcome type being $t_a$ approaches 1,
it follows that the utility of type $t$ at action $a$ approaches $v(t,a)-w(t,\chi(t_a,t_p))$ as $\alpha\rightarrow 1$ (for any $t_p\in T_p$).
This, together with Assumption~\ref{assumption:opt-actions} that $v(t,a_t)-w(t,\chi(t))> v(t,a')-w(t,\chi(t_{a'},t_p))$, implies that for large enough $\alpha$
the unique best reply of type $t$ is $a_t$.
Thus,  $(\sigma,\tau)$ is the unique equilibrium.
\end{proof}


Theorem~\ref{Claim:robustness} demonstrates the possibility of a tyranny of a majority that has no privacy concerns over a minority that does. This is quite interesting as the right to privacy has always been thought of as an aspect of individual liberty (this goes all the way back to the English Maxim that one's home is one's castle \citep{solove2006brief}). This phenomenon is a consequence of the externalities across types that is intrinsic to the notion of a perception equilibrium.  Thus, one must account for such dangers in the design of institutions for information sharing, even if participation is voluntary and information is non-verifiable.

Theorem~\ref{Claim:robustness} also has interesting implications in the context of a principal-agent setting. Consider a principal who introduces a mechanism that accounts for an agent that has no privacy concerns and contemplates the outcome in its unique separating equilibrium. Can she expect a similar outcome if she is wrong, and agents possibly have privacy concerns? The theorem suggests that if privacy concerns exists with a small (but positive) probability she can expect the same outcome, and hence the mechanism can be considered robust to a specific perturbation in the privacy hypothesis of the principal.

\subsection{Privacy and social welfare}\label{sec:general-model}
One common argument made against pro-privacy legislation is that it is detrimental to social welfare \citep[see][]{posner1981economics}. Two examples \citeauthor{posner1981economics} discusses are that concealing criminal records will lead to an inferior match in the labor market, and that confidentiality of debtors' information will distort  the credit market and lead to inefficient interest rates and risky loans. In these examples, privacy concerns as an intrinsic component of one's utility, and hence as a component of social welfare, have been ignored. Perhaps if these concerns would be accounted for then the conclusion would be that the material welfare decrease due to such privacy legislation is a worthwhile sacrifice for the benefit of privacy.

More precisely, suppose one can split total welfare into welfare from actions and welfare from privacy. Suppose also that the effect of privacy legislation is that actions are not observed, and
so perceptions are fixed at the optimal $\beta$ (that is, optimal when types are anonymity-seeking). The argument made by \cite{posner1981economics} and others is that privacy legislation reduces welfare from actions. However, privacy legislation may increase total welfare if welfare from privacy makes up the lost welfare from actions.

Our model suggests that privacy legislation may also increase total welfare by increasing welfare from actions. Recall that in some of the equilibria in Example \ref{Example1revisited} types pool and so one type takes an inferior action when actions are observed. Consider the same example with privacy regulation ensuring that no one can see the blog that the player reads. Clearly, each type would then read his optimal blog. So here, privacy legislation would lead to higher welfare from actions, and leave welfare from privacy unchanged (since types pool their actions when observed).

In fact, this last point is general: it is quite straightforward to see that for any one-player perception game, privacy legislation would always (weakly) increase total welfare whenever players are anonymity-seeking. It
(weakly) increases welfare from privacy---since under privacy legislation the DM obtains his most-preferred perception---and it also (weakly) increases welfare from actions, since under privacy legislation each type is free to choose his optimal action.\footnote{Clearly, BB's utility is ignored in this line of reasoning, and so it does not apply to \citeauthor{posner1981economics}'s examples. Those examples are, in fact, two-player games, where BB is the second player.}

However, this unequivocal benefit of privacy legislation no longer obtains in settings with more than one player. In the more general setting, it is possible that privacy legislation would decrease welfare from actions, and leave welfare from privacy unchanged.  As the example that follows will demonstrate, the lack of privacy legislation may yield higher welfare even if all members of society have privacy concerns. Thus, the conclusion advocated by \cite{posner1981economics}, that privacy legislation decreases total welfare,
may hold even if privacy confers its own intrinsic benefit on individuals, and even if that benefit is arbitrarily large.

The example is a perception game with two players. The formal model is provided in Appendix~\ref{sec:2-player}, which also includes Proposition~\ref{prop_existence_of_equilibrium} assuring the existence of an equilibrium.

\begin{example}\label{Example:2 player game}
%
Consider first the two-player Bayesian game depicted in Figure~\ref{fig:bayesian}, in which Nature draws a type profile from $\{u,d\}\times\{\ell,r\}$, and the payoffs are determined by the realized profile and the players' actions chosen from $\{U,D\}\times\{L,R\}$:

\begin{figure}[H]

\hspace*{\fill}
\begin{game}{2}{2}[$(u,\ell)$]
& $L$ & $R$\\
$U$ &$5,5$ &$0,0$\\
$D$ &$0,0$ &$0,0$
\end{game}\hspace*{20mm}
\begin{game}{2}{2}[$(u,r)$]
& $L$ & $R$\\
$U$ &$5,3$ &$0,4$\\
$D$ &$0,0$ &$0,1$
\end{game}\hspace*{\fill}%

\hspace*{\fill}
\begin{game}{2}{2}[$(d,\ell)$]
& $L$ & $R$\\
$U$ &$3,5$ &$0,0$\\
$D$ &$4,0$ &$1,0$
\end{game}\hspace*{20mm}
\begin{game}{2}{2}[$(d,r)$]
& $L$ & $R$\\
$U$ &$3,3$ &$0,4$\\
$D$ &$4,0$ &$1,1$
\end{game}\hspace*{\fill}%
\caption[]{A Bayesian game.%
}
\label{fig:bayesian}
\end{figure}%

Suppose also that player 1's (the row player) type is $u$ or $d$ with probability 0.5, and player 2's (the column player) type is $\ell$ or $r$ independently with
probability 0.5. In the unique Bayesian Nash equilibrium of this game, player 1 of type $u$ plays $U$ and player 1 of type $d$ plays $D$,
whereas player 2 of type $\ell$ plays $L$ and player 2 of type $r$ plays $R$. This yields the following interim expected utilities:

\begin{center}
  \begin{tabular}{| l || l | l | l | l | }
    \hline
    type: & $u$ & $d$ & $\ell$ & $r$ \\ \hline
    utility: & $2.5$ & $2.5$ & $2.5$ & $2.5$ \\
    \hline
  \end{tabular}
\end{center}

Consider now the perception game $G=(T,A,u,\beta)$, which is based on the Bayesian game of Figure~\ref{fig:bayesian} as follows:
The type space $T$ and action space $A$ are the same as in the Bayesian game, and the beliefs $\beta$ also correspond to the same commonly
known prior of 0.5 for each type. For each player $i$, the utility function is
$$u_i(t_1,t_2,a_1,a_2,\mu_i)=v_i(t_1,t_2,a_1,a_2)-w(\mu_i),$$
where $v(t_1,t_2,a_1,a_2)$ is the utility of player $i$ in the Bayesian game of Figure~\ref{fig:bayesian} and $w(\mu_i)$ is a privacy cost that depends
only on the belief $\mu_i$ of player $-i$ about $i$'s type (and denotes the probability that player $-i$ believes $i$ is of type $\ell$ or $u$).
Specifically, fix $w(\mu_i)\geq 0$ for all $\mu_i\in[0,1]$, with $w(0.5)=0$ and $w(0)=w(1)=1+\eps$ for some $\eps>0$. This corresponds to no
privacy cost when the posterior belief is 0.5, identical to the prior, and a nonnegative cost otherwise.

Now consider the strategy profile $\sigma$ in which $\sigma_1(u)=\sigma_1(d)=U$ and $\sigma_2(\ell)=\sigma_2(r)=L$. Also, fix the perception $\tau$ as $\tau_1(U)=\tau_2(L)=0.5$ and $\tau_1(D)=\tau_2(R)=0$. That is, the first pair of equalities is derived by consistency, whereas the latter equalities state that upon a deviation, the other player believes that the deviator is of type $R$ or $D$.

We claim that $(\sigma, \tau)$ is a perception equilibrium. To see this, consider for example player 1. Her expected utility under $\sigma$ is 5 under type $u$
and 3 under type $d$. Type $u$ clearly has no incentive to deviate. As for type $d$, a deviation will lead to a higher $v$ payoff of 4. However,
there will now be a privacy cost of $1+\eps$, since the belief of player 2 upon deviation will be 0. Thus, the total utility will be $4-1-\eps<3$,
so such a deviation is not profitable. A similar argument holds for player 2.

Finally, observe that the equilibrium $(\sigma, \tau)$ yields the following interim expected utilities:

\begin{center}
  \begin{tabular}{| l || l | l | l | l | }
    \hline
    type: & $u$ & $d$ & $\ell$ & $r$ \\ \hline
    utility: & $5$ & $3$ & $5$ & $3$ \\
    \hline
  \end{tabular}
\end{center}

Thus, even though the game $G$ consisted only of adding privacy costs to the Bayesian game of Figure~\ref{fig:bayesian},
the resulting equilibrium yields a strictly higher utility to both types of both players.
\end{example}

A related analysis of the interplay between social welfare and privacy policies appears in \cite{daughety2010public}. The main
conceptual difference between the two models is that in  \citeauthor{daughety2010public}'s model agents do not have privacy concerns but rather are concerned
about esteem, and wish to be perceived as the high type. Nevertheless, privacy policies do affect behavior in their model, and the effect may both enhance and diminish welfare.


\section{Concluding Remarks} \label{sec:conclusion}

Privacy concerns of individuals often depend on what the individuals expect others to learn from their own actions.
Thinking about privacy as the interplay between who the individual really is and what others know (or rather, believe) about him generates the following circle of reasoning
by him: Based on who I am, I will choose my action, which will then induce others to have a belief about who I am, which may then compel me to choose a different action. A perception game captures this circular reasoning and uses a fixed point---a perception equilibrium---to predict an outcome in such a setting.

Such equilibria, similarly to equilibria in signaling games, induce a full spectrum of privacy-related outcomes. At the extremes, they may be completely pooling, and hence others learn nothing about the individual and his privacy is kept intact, or they can be fully separating, in which case privacy is completely jeopardized.

Perception games provide a framework for the formal analysis of strategic settings in which actions and perceptions both play a role. Thus, they provide a unified framework for studying interactions among strategic players who care about the way others perceive them, which goes well beyond privacy concerns. In particular, perception games generalize models that have been used in the literature on social image and conspicuous consumption.

In this paper we restricted attention to a relatively simple class of perception games, where players' utilities depend on perceptions of their types but do not take perceptions over perceptions into account (in contrast to psychological games, where utilities may depend on a full hierarchy of beliefs). We also limited ourselves to a static, one-shot setting. We hope to study extensions of these restrictions in future work.

\subsection{Modeling privacy and optimal actions}\label{sec:modeling}

Our modeling choices---that only perceptions matter, that the setting is static, and that privacy concerns are a direct argument of the utility
function---are partially made for tractability, but additionally have some motivation in the context of privacy.
Let us first pause to think about the nature of privacy concerns. What is it that individuals do not want others to know, and why? It seems that such concerns arise at three different levels:
\begin{itemize}
\item
An embarrassment derived directly from others observing one's actions, even if that sheds no new light on how the individual is perceived. For example, privacy is important in public restrooms in many Western societies, and many people feel awkward if others watch them in the act.%
  \footnote{Similarly, although it is commonly assumed that consumption of adult entertainment is very widely spread, most consumers would feel uncomfortable if observed in the act of consumption \citep[see][]{edelman2009markets}.}
\item
A concern about revealing private information that can lead others to have the upper hand in some future (specified or unspecified) interaction. For example, concealing one's juvenile crime record is important for securing the impartiality of future potential employers.
\item
A concern about inferences others might make about an individual by observing his actions. Such inferences might relate to various aspects of his personality, such as his attitude towards risk, altruism, consumption preferences, etc. For example, in certain countries, men, who might otherwise appreciate the comfort of a skirt, may not wear skirts in public due to concerns regarding what people would deduce about their masculinity.%
\footnote{Some argue that the fundamental concern that is captured by this argument is in fact a proxy for concerns that are related to future payoffs in some unknown future interaction. Hence, the distinction between this argument and the previous one may boil down to how carefully the individual models his future interactions.}

\end{itemize}

In the standard economic and game-theoretic modeling paradigm, these three levels of concern translate to:
\begin{itemize}
\item
Disutility from others observing one's actions, which can be encoded into the utility function. This can be modeled in the context of a simple game with complete information.
\item
The strategic options available to others in the future due to finer information. This can be captured in a model of a repeated game
\citep[as in][]{mailath2006repeated}.
\item
What others learn about one's type, which can be captured in a model of incomplete information.
\end{itemize}

In this paper we study how privacy concerns impact individuals' decisions. In particular, our work is motivated by the two last levels of concern mentioned above. Our utility function in a single shot game, which accounts for perceptions, can be thought of as a reduced-form analysis of a repeated game whenever future interactions (future stage games) are ambiguous and will take different forms in
different circumstances. In that case players have no clear structure of the repeated game and reduce the implications of choosing one action or another to its effect on the perception. Another interpretation of the model is in the context of the last level above. Thinking of our model in that context captures the abstract idea that privacy, interpreted as the lack of additional information on one's type, is intrinsically important to the DM.

\subsection{Who is Big Brother?}

The passive player, BB, should be thought of in one of two ways. He may either be a real player, or, alternatively, he can merely be a figment of the DM's imagination. Technically this distinction has no bite, but conceptually it matters. The former interpretation sits well with our motivation for the model as a reduced-form model for some elaborate setting where BB takes future actions that are material to the DM. However, given the vagueness of how the information will be used, which course of action (or actions) will be taken, and how this will actually impact the DM's utility, one folds all of this into a perception of BB over the type of DM. In technical terms, the perception is a sufficient statistic for the future evolution of the game.

The latter interpretation is more in line with a different view of what privacy concerns are all about. In this case one should think of privacy not as a means for getting future benefits but rather as a virtue that underpins human dignity and freedom. In this case, what the DM finds unsettling is the fact that his type is exposed and so people can learn who he is.

\subsection{Modeling esteem vs.\ modeling privacy}\label{subsection percpetion vs esteem}

As discussed in the introduction the general approach underlying perception games, and specifically our modeling choice for capturing privacy concerns, is closely related to the literature on social image and conspicuous consumption. In both cases a player's utility depends, inter alia, on how he is perceived. However, a few distinctions are worth noting:
\begin{itemize}
\item
In our setting the DM does not want others to learn anything about his type. Compare this with the literature where in all models (as far as we know) there is always some `ideal' type (or more generally, types are ordered) and all players, independently of their type, would like to be perceived to be as close as possible to the ideal type.%
 \footnote{The `ideal' type may be referred to as the benevolent type, the wealthy type, the altruistic type, the bold type, etc.}
Furthermore, this notion of an ideal type is the same for all agents.%
\footnote{The recent work of \cite{friedrichsen2013image} departs from this strong homogeneity assumption.} In fact, if others learn that the DM's type is an ideal one then he is happy with that. This renders types that are close to ideal as averse to privacy in our sense, and is quite different from our notions of privacy and identity concerns.
\item
A technical distinction worth noting is that in our general model of perception games, perceptions are captured by distributions over types. However, in most models on social image, the type set is modeled as an interval on the real line and a perception is simply captured by a single type (and not a distribution). This is not a limiting issue in the context of a separating equilibrium.
On the other hand, the primitive we allude to is the distribution over beliefs, and we do not restrict the discussion to utility functions which treat similarly all beliefs with the same expectation. To demonstrate the limit of the `expected type' approach consider the imposed  equivalence between the following two scenarios: (1) Big Brother does not know what one's political views are and assigns a uniform probability to any of the views on the interval spanning from extreme left to extreme right, and (2) Big Brother is confident the individual is an extremist, either left or right, and assigns each of these cases a probability of $0.5$.
\end{itemize}





\newpage
\begin{center}
\begin{Large}
\textbf{Appendix}
\end{Large}
\end{center}

\appendix

\section{Two-Player Perception Games}\label{sec:2-player}
A {\it two-player perception game} is a tuple $(T_i,A_i,\beta_i,u_i)_{i=1}^2$ defined as follows. For each player $i\in\{1,2\}$,
\begin{itemize}
\item
$T_i$ is a finite type space,
\item
$A_i$ is a finite action space,
\item
$u_i:T_i \times T_{-i} \times A_i \times A_{-i} \times \Delta(T_i) \to \R$.
\item
$\beta_i: T_i \to \Delta(T_{-i})$ is a {\em belief} for player $i$.
\end{itemize}

Note that the two-player perception game reduces to our initial model of a perception game, hereinafter known as a {\it single-player perception game}, whenever $T_2$ and $A_2$ are singleton sets.

Similarly to the single-player perception game, a {\em strategy} for player $i$ is a function $\sigma_i: T_i \to \Delta(A_i)$.
A perception for player $i$ is a function $\tau_i:T_i \times T_{-i} \times A_i \to \Delta(T_i)$, which represents what $i$ believes $-i$'s belief (of some type $t_{-i}$) is over $i$'s own type.

Given a tuple $(\beta_i,\tau_i)_{i=1}^2$, the expected payoff to player $i$, of a given type $t_i$ when players use strategies $(\sigma_j)_{j=1}^2$, is
$U_i((\sigma_j)_{j=1}^2,\beta_i,\tau_i)(t_i)$ (also denoted $U_i((\sigma_j)_{j=1}^2, \tau_i)(t_i)$ or $U_i((\sigma_j)_{j=1}^2)(t_i)$   when the beliefs and perceptions are clear), computed as follows:
\begin{align*}
&U_i((\sigma_j)_{j=1}^2)(t_i)  \\
&=\sum_{a\in A_i,t_{-i}\in T_{-i},b\in A_{-i}}\sigma_i(t_i)(a)\cdot \beta_i(t_i)(t_{-i})\cdot \sigma_{-i}(t_{-i})(b)\cdot u_i(t_i,t_{-i},a,b,\tau_i(t_i,t_{-i},a)).
\end{align*}

Let us consider the perceptions of player $i$ were she aware of the beliefs of player $-i$. Let $P_{t_i}(t_{-i},b) =
\beta_i(t_i)(t_{-i}) \cdot \sigma_{-i}(t_{-i})(b)$
be the probability that $i$, of type $t_i$, assigns to $-i$ being of type $t_{-i}$ and to this type performing the action $b$.
Thus, $P_{t_i}(b) = \sum_{t_{-i}\in T_{-i}}P(t_{-i},b)$ is the probability she assigns to $-i$ performing the action $b$ and $P_{t_i}(t_{-i}| b)= \frac{P_{t_i}(t_{-i},b)}{P_{t_i}(b)}$ (where $\frac{0}{0}$ is interpreted as $0$) is the conditional probability that $i$, of type $t_i$, assigns to $-i$ being of type $t_{-i}$ upon seeing the action $b$ of $-i$. Symmetrically, let $P_{t_{-i}}(\bar t_i| a)$ be the conditional probability that $-i$, of type $t_{-i}$, assigns to $i$ being of type $\bar t_i$ upon seeing the action $a$ of $i$.

Player $i$ is {\em consistent} if for all $t_{-i}$, $t_i$, and 
$a\in \cup_{\tilde t_i \in \supp(\beta_{-i}(t_{-i}))}\supp(\sigma_i(\tilde t_i))$, it holds that
$\tau_i(t_i,t_{-i},a)(\bar t_i) = P_{t_{-i}}(\bar t_i|a)$ for all $\bar t_i$. In words, the ex post perception of one player over the type of the other player must be Bayesian consistent.

\bdefi
A {\em perception equilibrium} is a tuple 
$(\sigma_i,\tau_i)_{i=1}^2$ such that both players are consistent and for any $i$ and $t_i\in T_i$,
the strategy
$\sigma_i$ is a best reply: $U_i((\sigma_i)_{i=1}^2)(t_i) \ge
U_i(\hat\sigma_i,\sigma_{-i})(t_i) \ \forall \hat\sigma_i:T_i \to \Delta(A_i)$.
\edefi

As with the single-player perception game, we assume continuity in the last argument of the utility function:
\begin{assumption}\label{assumption:continuity_in_games}
For each $i$, the utility function $u_i:T_i \times T_{-i} \times A_i \times A_{-i} \times \Delta(T_i) \to \R$ is continuous in the last argument.
\end{assumption}

\begin{proposition}\label{prop_existence_of_equilibrium}
A perception equilibrium exists in every two-player perception game that satisfies Assumption \ref{assumption:continuity_in_games}.
\end{proposition}

\section{Perception games as a variant of signaling games}\label{sec:signaling games analogy}

As suggested in the introduction, perception games are a variant of signaling games. The initial analogy is that of a signaling game played between two players, with the DM serving as a sender and BB serving as a receiver whose utility is determined by a proper scoring rule. Let us examine this analogy and its implication for equilibrium existence.

In a signaling game the receiver has one single belief in each of his information sets. In particular, as the receiver does not know the type of the sender the model dictates that the receiver's belief is independent of the type of the sender. In our setting this is not required but is an outcome of the equilibrium concept, and even then holds only for non-zero probability events. Specifically, off the equilibrium path different types of sender (DM) may have different perceptions.

Thus, in order to take the exercise of embedding perception games within the framework of signaling games we must go beyond a single sender and a single receiver game. The correct model is one of a signaling game with multiple receivers, one for each type of the sender. The sender of type $t$ sends a single message to $T$ receivers and only cares about the ``action'' of receiver $t$. The perception associated with type $t$ of the DM in a perception game corresponds to the belief of receiver $t$ in the signaling game. Both perception equilibria and Perfect Bayesian Equilibria (PBE) require that perceptions/beliefs be ``correct" in equilibrium, i.e., on the equilibrium path, and this is the consistency requirement in both perception equilibria and in the PBE of the corresponding signaling game. Thus, the PBE of the resulting signaling game will perfectly match up with a perception equilibrium of the original perception game. Note that in such a multi-receiver signaling game the diversity of off-equilibrium perceptions sits well with the interpretation of BB as a real player (or players).

Next, one can ask whether this implies the equilibrium existence result. To the best of our understanding it does not. The reason is that the resulting signaling game is not a finite game, since the receiver has an infinite (though compact) action space: the set of all probability distributions over sender types. We were not able to find general existence results for PBE in this class of games, and in fact \cite{manelli1996cheap}
demonstrates a signaling game with no PBE.%
\footnote{Theorem 1 in \cite{manelli1996convergence}
 provides an existence result for the special case of binary sender types. It does not extend beyond binary types as it assumes that the action set of the receiver is a real number, which can correspond to a probability in the binary case but not to distributions over more than two actions.}

\section{Proofs of Propositions~\ref{prop_existence} and~\ref{prop_existence_of_equilibrium}}\label{sec:proofs}
The proof, not surprisingly, hinges on Kakutani's fixed-point theorem. As usual, this involves defining an appropriate best-response correspondence, showing that it is upper-hemicontinuous and convex-valued, and then invoking Kakutani's theorem. The crux of our proof lies in fixing the perception $\tau$ as a function of the strategy, such that the induced best-response correspondence is indeed
upper-hemicontinuous. More specifically, in defining the best-response correspondence we construct a particular consistent $\tau^\sigma$ for each
strategy profile $\sigma$ in the domain, for which we find the set of best-responses.

We will prove Proposition~\ref{prop_existence_of_equilibrium}. Proposition~\ref{prop_existence} is simply a corollary,
with $\abs{T_2}=\abs{A_2}=1$.

Fix a perception game $(T_i,A_i,\beta_i, u_i)_{i=1}^2$, a strategy profile $\sigma=(\sigma_1,\sigma_2)$, a player $i\in\{1,2\}$, and a type $t\in T_i$
of player $i$.\footnote{We use $t\in T_i$ rather than $t_i$ to slightly shorten notation.}
Let ${\cal T}_t^\sigma$
be the set of all perceptions for type $t$ that are consistent given $\sigma$. Formally,
\begin{align*}
{\cal T}_t^\sigma &= \Large\{\tau_i:T_i\times T_{-i}\times A_i \to\Delta(T_i): \tau_i(t,t_{-i},a)(\bar t) = P_{t_{-i}}(\bar t|a) \\
&~~~~~~ \forall \bar t \in T_i,
\forall t_{-i}\in T_{-i}, \forall a\in \cup_{\tilde t_i \in \supp(\beta_{-i}(t_{-i}))}\supp(\sigma_i(\tilde t_i))\Large\}.
\end{align*}

Let $B_{t}^\sigma\subseteq A$ be defined as follows:
$$
B_{t}^\sigma = \large\{a\in A_i:~ \exists \mu_i \in {\cal T}_t^\sigma ~\mbox{s.t.} \
U_i(\sigma_i^{a},\sigma_{-i},\beta_i,\mu_i)(t)\geq U_i(\sigma^{a'},\sigma_{-i}\beta_i,\mu_i)(t) \ \forall a' \in A \large\},
$$
where $\sigma_i^a$ is the pure strategy of player $i$ that puts weight 1 on the action $a$.
$B_{t}^\sigma$ consists of all actions $a\in A_i$ that are weakly optimal for a player $i$ of type $t$ for some consistent perception of type $t$. In particular, any $a\not\in B_t^\sigma$ is never optimal as long as perceptions are consistent.

\begin{claim}\label{claim_exists_privacy_concern}
There exists
$\tau_{t}^{\sigma} \in {\cal T}_t^\sigma$ for which $U_i(\sigma^a,\sigma_i,\beta_i,\tau_{t}^{\sigma})(t)=U_i(\sigma^{a'},\sigma_i,\beta_i,\tau_{t}^{\sigma})(t)$ for all $a,a'\in B_t^\sigma$.
\end{claim}

\begin{proof}
Assumption \ref{assumption:continuity} implies that for any $a\in B_t^\sigma$,
the set $U_i^a\eqdef\{U_i(\sigma^a,\sigma_i,\beta_i,\mu_i)(t): \mu_i \in {\cal T}_t^\sigma \}$ is a closed interval in $\R$. Furthermore, by the construction of $B_t^\sigma$, the pairwise intersection $U_i^a\cap U_i^{a'}\neq \emptyset$ for any
$a,a'\in B_t^\sigma$. Thus, the intersection of all such intervals is nonempty. Let  $p\in \R$ denote some point in the intersection of all of the intervals. For each $a\in B_t^\sigma$,
this point is attained at some $\mu^a=\mu^a(t,t_{-i},a) \in {\cal T}_t^\sigma$, so
\begin{align*}
p&=U_i(\sigma^a,\sigma_{-i},\beta_i,\mu^a)(t) \\
& \sum_{t_{-i}\in T_{-i},b\in A_{-i}} \beta_i(t)(t_{-i})\cdot \sigma_{-i}(t_{-i})(b)\cdot u_i(t,t_{-i},a,b,\mu^a(t,t_{-i},a)).
\end{align*}


Define $\tau_{t}^{\sigma}$ as $\tau_{t}^{\sigma}(t,t_{-i},a)\eqdef \mu^a(t,t_{-i},a)$ for every $t_{-i}\in T_{-i}$ and $a\in A_i$, and note that $\tau_{t}^{\sigma} \in {\cal T}^\sigma_t$ as well. Then for each $a\in B_t^\sigma$ it holds that
\begin{align*}
U_i(\sigma^a,\sigma_{-i},\beta_i,\tau_{t}^{\sigma})(t)
&= \sum_{t_{-i}\in T_{-i},b\in A_{-i}} \beta_i(t)(t_{-i})\cdot \sigma_{-i}(t_{-i})(b)\cdot u_i(t,t_{-i},a,b,\tau_t^\sigma(t,t_{-i},a))\\
&=p
\end{align*}
as desired.
\end{proof}

For any strategy profile $\sigma$ and player $i$ fix the perception $\tau^\sigma_i$ defined as
$\tau_{i}^{\sigma}(t,t_{-i},a)\eqdef \tau_t^\sigma(t,t_{-i},a)$ for every $t\in T_i$, $t_{-i}\in T_{-i}$ and $a\in A_i$  satisfying
$U_i(\sigma^a,\sigma_i,\beta_i,\tau_{i}^{\sigma})(t)=U_i(\sigma^{a'},\sigma_i,\beta_i,\tau_{i}^{\sigma})(t)$ for all $a,a'\in B_t^\sigma$. Such a perception is guaranteed to exist by Claim \ref{claim_exists_privacy_concern}.  Define the following correspondence
$\br^i_t$:
$$
\br^i_{t}(\sigma)
=\arg\max_{\mu \in\Delta(A_i)}\left\{U_i\left(\mu, \sigma_{-i},\beta_i,\tau_i^{\sigma}\right)(t)\right\}.
$$

\begin{claim}\label{claim_br_convex}
$\br^i_t(\sigma)$ is nonempty and convex for every 
$\sigma=(\sigma_1,\sigma_2)$, where $\sigma_i:T_i\to A_i$.
\end{claim}

\begin{proof}
Observe that for any $\sigma'(t),\sigma''(t)\in\Delta(A)$, a fixed  perception  $\tau^\sigma$, and $\lambda\in(0,1)$, it holds that
\begin{align*}
&U_i\left(\lambda\sigma'(t)+(1-\lambda)\sigma''(t), \sigma_{-i},\beta_i,\tau^{\sigma}\right)(t)\\
&~~~~~=\lambda U_i\left(\sigma'(t), \sigma_{-i},\beta_i,\tau^{\sigma}\right)(t)+(1-\lambda)U_i\left(\sigma''(t), \sigma_{-i},\beta_i,\tau^{\sigma}\right)(t).
\end{align*}
In particular, this implies that $U_i$ is concave in its first argument, which in turn implies that the set $\br^i_t(\sigma)$ is nonempty and convex for every $\sigma$.
\end{proof}

Consider the following correspondence 
defined by
$$\hat\br_{t}^{\sigma_{-i}}(\nu) =\arg\max_{\sigma'(t)\in\Delta(A)}\left\{U_i\left(\sigma'(t),\sigma_{-i},\beta_i,\nu\right)(t)\right\}.$$

Note in particular that $\hat\br_t^{\sigma_{-i}}(\tau_i^\sigma) = \br_t^i(\sigma)$.
As with $\br_t(\sigma)$, the set $\hat\br_t(\nu)$ is also nonempty
and convex.

\begin{claim}\label{claim_3}
 If $\mu_t \in {\cal T}_t^\sigma$ and $a\in\hat\br_t^{\sigma_{-i}}(\mu_t)$, then $a\in\br_t^i(\sigma)$.
\end{claim}

\begin{proof}
The conditions in the claim directly imply that $a\in B_t^\sigma$. 
By Claim~\ref{claim_br_convex} there exists some $b\in\br_t^i(\sigma)$. Note that this implies that $b\in B_t^\sigma$.
By construction, since both $a$ and $b$ are in $B_t^\sigma$, it must be the case that $U_i(\sigma^a,\sigma_{-i},\beta_i,\tau^\sigma_i)(t) = U_i(\sigma^b,\sigma_{-i},\beta_i,\tau^\sigma_i)(t)$.
But $b\in\br_t^i(\sigma)$ then implies that $a\in\br_t^i(\sigma)$.
\end{proof}

\begin{lemma}\label{lemma:uhc}
$\br_t^i(\sigma)$ is upper-hemicontinous.
\end{lemma}

\begin{proof}
Fix a sequence $\sigma^n\rightarrow \sigma$ with $a\in\br_t^i(\sigma^n)$ for all $n$.
Due to the convexity of $\br_t^i(\sigma)$ guaranteed by Claim~\ref{claim_br_convex}, a sufficient condition for upper-hemicontinuity is that $a\in\br_t^i(\sigma)$.

Let $n$ be sufficiently large to guarantee that for all $t_{-i}\in T_{-i}$,
$$\{a\in \cup_{\tilde t_i \in \supp(\beta_{-i}(t_{-i}))}\supp(\sigma_i(\tilde t_i))\} \subseteq \{a\in \cup_{\tilde t_i \in \supp(\beta_{-i}(t_{-i}))}\supp(\sigma^n_i(\tilde t_i))\}.$$
This is possible, since the action set is finite and so the support set for $\sigma^n$ must be equal that of $\sigma$ for large enough $n$.
Consider the sequence of player $i$ perceptions $\left\{\tau_i^{\sigma^n}\right\}_{n\geq 1}$. By Bolzano--Weierstrass, this sequence has
a convergent subsequence $\left\{\tau_i^{\sigma^m}\right\}_{m\geq 1}$, where
$$q_t\eqdef \lim_{m\rightarrow \infty} \tau_i^{\sigma^m}.$$
Observe that since $\tau_i^{\sigma^m}\in{\cal T}_t^{\sigma^n}$ for all $m$,  by the assumption
of $n$ being large enough it must be the case that $q_t\in {\cal T}_t^\sigma$. Since $a\in\hat\br_t^{\sigma_{-i}}(\tau_t^{\sigma^m})$, by continuity $a\in\hat\br_t^{\sigma_{-i}}(q_t)$.
But then Claim~\ref{claim_3} implies that $a\in\br_t^i(\sigma)$, as required.
\end{proof}

\noindent We are finally ready to prove Proposition \ref{prop_existence}:

\begin{proof}
By Lemma~\ref{lemma:uhc}, $\br_t^i(\sigma)$ is upper-hemicontinuous. This implies that the correspondence $\br$
defined as $\br(\sigma)=\left\{\br_{t_i}^i(\sigma)\right\}_{i\in\{1,2\},t_i\in T_i}$, is also convex-valued and
upper-hemicontinuous,
 so by Kakutani's theorem a fixed point $\sigma^*$ exists.
It is easy to verify that $(\sigma^*,\tau_1^{\sigma^*},\tau_2^{\sigma^*})$ is a perception equilibrium.
\end{proof}

\section{Continuity is necessary for equilibrium existence}\label{sec:continuity-necessary}

The following example of a single-player perception game is a (quite artificial) variant of Example~\ref{example:Coventry} which violates Assumption \ref{assumption:continuity} and  has no perception equilibrium.

\begin{example}
The active player, Alice, is of type $t\in\{\ell,r\}$ with $\beta(1)=\beta(0)=0.5$, and must choose between actions $\{L,R\}$, with a utility of 1 if $t=a$ (the uppercase of the type is equal to the action) and 0 otherwise.
However, suppose that Alice incurs an additional disutility as follows. Let $p \in [0,1]$ denote the probability that the inactive player assigns to Alice type $r$.
Then both types incur a disutility of $-2$ whenever $p=0.5$, type $r$ incurs a disutility of $-2$ when $p=1$, and type $\ell$ incurs a disutility of $-2$ whenever $p=0$. The disutility is 0 otherwise.

For any pure and separating strategy profile, both players obtain a utility of at most $1-2$, whereas a deviation would yield 0 and is hence profitable.
Suppose the strategy profile is pure and pooling, with both types playing $L$. Then type $\ell$ obtains utility $-1$, and type $r$ obtains utility $-2$. But regardless of
the perception under action $r$, type $r$ will strictly benefit from deviating. Finally, as in Example~\ref{example:Coventry}, no mixed strategy can be optimal.
\end{example}
Although a proper definition is not given, there is no $\eps$-perception equilibrium in this example for any small enough $\eps$.

Note that the utility functions in the example, while not continuous in the last argument, are lower semicontinuous, and so a weaker assumption to that effect
does not suffice for equilibrium existence. To see that upper semicontinuity alone is also not sufficient, modify the payoffs of the example above as follows: Fix a small
$\eps>0$, and suppose both types incur a disutility of $-2$ whenever $p \in (.5-\eps, .5+\eps)$, type $r$ incurs a disutility of $-2$ when $p > 1-\eps$, and type $\ell$ incurs a disutility of $-2$ whenever $p < \eps$. As above, the disutility is $0$ otherwise. With these modifications, the utility functions are upper semicontinuous. However, the same argument for lack of an equilibrium as in the example still goes through.

\paragraph{Acknowledgements} Gradwohl gratefully acknowledges the support of NSF award \#1216006. Smorodinsky gratefully acknowledges the support of ISF grant 2016301, the joint Microsoft-Technion e-Commerce Lab, Technion VPR grants and the Bernard M. Gordon Center for Systems Engineering at the Technion. We also thank Eddie Dekel, Jana Friedrichsen, Ehud Kalai, Gil Kalai, Birendra Rai, Juuso Valimaki and three anonymous reviewers for insightful comments.

\bibliography{privacy-aware-decision-making-bib}

\end{document}